\documentclass[journal, draftclsnofoot,oneside, onecolumn,11pt, letterpaper]{IEEEtran}

\ifCLASSINFOpdf
\else
\fi
\usepackage{pslatex}
\usepackage{amsfonts,color,morefloats}
\usepackage{amssymb,amsmath,latexsym,amsthm}
\usepackage{bbm}
\usepackage{amsfonts}
\usepackage{mathrsfs}
\usepackage{amsthm}
\usepackage{latexsym}
\usepackage{color}
\usepackage{dcolumn}
\usepackage{multicol}
\usepackage{extarrows}
\usepackage{amsmath}

\allowdisplaybreaks[4]

\newtheorem{theorem}{Theorem}
\newtheorem{lemma}[theorem]{Lemma}

\newtheorem{corollary}[theorem]{Corollary}

\newtheorem{definition}{Definition}

\newtheorem{remark}{Remark}

\begin{document}
%
\title{4-Adic Complexity of Interleaved  Quaternary Sequences
}
%
%
%

\author{Shiyuan~Qiang,
        Xiaoyan~Jing,
        Minghui~Yang,
        and~Keqin~Feng
\thanks{Keqin Feng was supported by
the National Natural Science Foundation of China (NSFC) under Grant 12031011.}
\thanks{Shiyuan Qiang is with the Department of Applied Mathematics, China Agricultural University, Beijing 100083, China (e-mail: qsycau\_18@163.com).}
\thanks{Xiaoyan Jing is with School of Mathematics, Northwest University, Xi'an 710127, China (e-mail: jxymg@126.com).}
\thanks{Minghui Yang is with State Key Laboratory of Information Security, Institute of Information Engineering, Chinese Academy of Sciences, Beijing 100093, China (e-mail:  yangminghui6688@163.com).}
\thanks{Keqin Feng is with the Department of Mathematical Sciences, Tsinghua University, Beijing 100084, China (e-mail: fengkq@tsinghua.edu.cn).}
}

%
%

\markboth{}%
{Shell \MakeLowercase{\textit{et al.}}: Bare Demo of IEEEtran.cls for IEEE Journals}
%



\maketitle


%
\IEEEpeerreviewmaketitle

\begin{abstract}
Tang and Ding \cite{X. Tang} present a series of quaternary sequences $w(a, b)$ interleaved by two binary sequences $a$ and $b$ with ideal autocorrelation and show that such interleaved quaternary sequences have optimal autocorrelation. In this paper we consider the 4-adic complexity $FC_{w}(4)$ of such quaternary sequence $w=w(a, b)$. We present a general formula on $FC_{w}(4)$, $w=w(a, b)$. As a direct consequence, we obtain a general lower bound $FC_{w}(4)\geq\log_{4}(4^{n}-1)$ where $2n$ is the period of the sequence $w$. By taking $a$ and $b$ to be several types of known binary sequences with ideal autocorrelation  ($m$-sequences, twin-prime, Legendre,  Hall sequences and their complement, shift or sample sequences), we compute the exact values of $FC_{w}(4)$, $w=w(a, b)$ and show that in most cases $FC_{w}(4)$ reaches or nearly reaches the maximum value $\log_{4}(4^{2n}-1)$. Our results show that the 4-adic complexity of the quaternary
sequences defined in \cite{X. Tang} are large enough to resist the attack of the rational approximation
algorithm.

\end{abstract}

\begin{IEEEkeywords}
quaternary sequences, 4-adic complexity, interleaved sequences, ideal autocorrelation
\end{IEEEkeywords}

%
\IEEEpeerreviewmaketitle

\section{Introduction}\label{section1}
Binary and quaternary sequences with good autocorrelation have been widely used in digital communication and cryptography (\cite{T. W. Cusick,S. W. Golomb}). For a binary sequence $s=\{s_\lambda\}_{\lambda=0}^{n-1}$ with period $n$, $s_\lambda\in \{0, 1\}$, the autocorrelation function of $s$ is defined by
$$A_s(\tau)=\sum_{\lambda=0}^{n-1}(-1)^{s_\lambda+s_{\lambda+\tau}}\in \mathbb{Z} \ (0\leq\tau\leq n-1),$$
$A_s(0)=n$ and
in many applications we need that the absolute values $|A_s(\tau)|$ $(1\leq\tau\leq n-1)$ are small. It is easy to see that $A_s(\tau)\equiv n\pmod 4$ for all $0\leq\tau\leq n-1$.
We call a binary sequence $s$ with ideal autocorrelation if $n\equiv3\pmod 4$ and $A_s(\tau)=-1$ for all $1\leq\tau\leq n-1$. Many series of binary sequences with ideal autocorrelation have been found (see \cite{X. Tang}, Section (\uppercase\expandafter{\romannumeral2}, $C$)). For all known binary sequences with ideal autocorrelation, their period $n$ is $2^{m}-1$ $(m\geq2)$, prime number $p$ or product of twin primes $p$ and $p+2$.

For a quaternary sequence $s=\{s_\lambda\}_{\lambda=0}^{n-1}$ with period $n$, $s_\lambda\in \{0, 1, 2, 3\}$, the autocorrelation function of $s$ is defined by
$$A_s(\tau)=\sum_{\lambda=0}^{n-1}\omega^{s_{\lambda+\tau}-s_\lambda}\in {\mathbb{Z}[\sqrt{-1}]},$$
where $\omega=\sqrt{-1}$ is a complex primitive $4$-th root of unity. In many applications we need the absolute values $|A_s(\tau)|$ for all out-of-phase autocorrelation $A_s(\tau)$ $(1\leq\tau\leq n-1)$ are small.

The feedback with carry shift register (FCSR) is a kind of
nonlinear sequences generator.  The $d$-adic complexity $FC_s(d)$ measures the smallest length of FCSR which generates the sequence $s$ over $\mathbf{Z}/(d)$. In cryptographic applications, sequences should not only have low autocorrelation values, but also have larger 4-adic complexity. Particularly, a quaternary sequence $s$ with period $N$ can be generated by a feedback shift register with 4-adic carry and length at least $\lceil FC_{s}(4)\rceil$, and the 4-adic complexity of such sequences  should exceed $\frac{N-16}{4}$ in order to resist the attack of the rational approximation algorithm, see \cite{K4,K1}.

\begin{definition}\label{def1}
For a quaternary sequence $s=\{s_\lambda\}_{\lambda=0}^{n-1}$ with period $n$, $s_\lambda\in \{0, 1, 2, 3\}$, the 4-adic complexity of $s$ is defined by
$$FC_{s}(4)=\log_{4}\frac{4^{n}-1}{d},$$
where $d=\gcd (S(4), 4^{n}-1)$, $S(4)=\sum_{\lambda=0}^{n-1}s_{\lambda}4^{\lambda}\in \mathbb{Z}$.
\end{definition}

Several quaternary sequences with good autocorrelation have been constructed (\cite{jang,kim,Li,Su}). Tang and Ding \cite{X. Tang} construct a series of quaternary sequences with period $2n$ and optimal autocorrelation by interleaving a pair of binary sequences with ideal autocorrelation and period $n\equiv3\pmod 4$. In this paper we compute the 4-adic complexity of such quaternary sequences.

In Section \ref{sec2} we introduce the interleaving construction of quaternary sequence with period $2n$ by a pair of binary sequences with the same odd period $n$, compute the value of $S(4)\pmod {4^{2n}-1}$ for such quaternary sequences, and present a general result on the 4-adic complexity of quaternary sequences interleaved by a pair of binary sequences with ideal autocorrelation. In Section \ref{sec33} we determine the exact value of the 4-adic complexity for quaternary sequences interleaved by some particular types of binary sequences with ideal autocorrelation ($m$-sequences, Legendre sequences, twin-prime sequences and Hall sequences). The last section is a conclusion.


\section{Interleaved Quaternary Sequences}\label{sec2}

In this section we introduce the interleaving construction of quaternary sequences with period $2n$ by a pair of binary sequences with the same odd period $n$. For general theory and applications of interleaved sequences we refer to Gong \cite{G. Gong}. In this paper we use the following simple form.

Let $n$ be an odd integer, $n\geq3$, $a=\{a_\lambda\}_{\lambda=0}^{n-1}$ and $b=\{b_\lambda\}_{\lambda=0}^{n-1}$ be a pair of binary sequences, $a_\lambda, b_\lambda\in \{0, 1\}$, with the same period $n$. Consider the following two $n\times2$ matrices over $\mathbb{F}_{2}=\{0, 1\}$ $(1+1=0)$
$$M_{a}=\left(
  \begin{array}{cc}
    a_{0} & a_{\frac{n+1}{2}}  \\
    a_{1} & a_{\frac{n+1}{2}+1} \\
    \vdots & \vdots  \\
    a_{n-1} & a_{\frac{n+1}{2}+n-1}  \\
  \end{array}
\right),
~~M'_{b}=\left(
  \begin{array}{cc}
    b_{0} & 1-b_{\frac{n+1}{2}}  \\
    b_{1} & 1-b_{\frac{n+1}{2}+1}  \\
    \vdots & \vdots  \\
    b_{n-1} & 1-b_{\frac{n+1}{2}+n-1}  \\
  \end{array}
\right).
$$
Then we get an interleaved binary sequence with period $2n$ from the matrix $M_{a}$
\begin{align}\label{E1}
c& =\{c_{i}\}_{i=0}^{2n-1}\notag\\
&= (c_{0}, c_{1}, c_{2}, c_{3}, \cdots, c_{2n-2}, c_{2n-1})\notag\\
&= (a_{0}, a_{\frac{n+1}{2}}, a_{1}, a_{\frac{n+1}{2}+1}, \cdots, a_{n-1}, a_{\frac{n+1}{2}+n-1}(=a_{\frac{n-1}{2}})).
\end{align}

By the Chinese Remainder Theorem we have the isomorphism of rings
$$f: \mathbb{Z}_{2n}\cong \mathbb{Z}_{2}\oplus \mathbb{Z}_{n},  ~~~i(\bmod{2n})\mapsto(i (\bmod 2), i(\bmod n)).$$
For $(\mu, \lambda)\in \mathbb{Z}_{2}\times \mathbb{Z}_{n}$, we have $f^{-1}((\mu, \lambda))=n\mu+(1-n)\lambda$ ($0\leq\mu\leq1, 0\leq\lambda\leq n-1$). It is easy to see from the definition (\ref{def1}) of the sequence $c$ that for $i=n\mu+(1-n)\lambda$,
\begin{align*}
 c_{i}
 &= \left\{ \begin{array}{ll}
a_{\frac{i}{2}}=a_{\lambda(\frac{1-n}{2})}, & \textrm{if $2\mid i$ ~(namely, $\mu=0$)};\\
a_{\frac{n+1}{2}+\frac{i-1}{2}}=a_{\lambda(\frac{1-n}{2})}, & \textrm{if $2\nmid i$ ~(namely, $\mu=1$)}.
\end{array} \right.
\end{align*}
(Since for $\mu=1$, $\frac{n+1}{2}+\frac{i-1}{2}=\frac{n+1}{2}+\frac{n+(1-n)\lambda-1}{2}=n+\lambda(\frac{1-n}{2})\equiv \lambda(\frac{1-n}{2})\pmod n$).

Similarly we have another interleaved binary sequence with period $2n$ from the matrix $M'_{b}$
\begin{align*}
d& =\{d_{i}\}_{i=0}^{2n-1}\notag\\
&= (d_{0}, d_{1}, d_{2}, d_{3}, \cdots, d_{2n-2}, d_{2n-1})\notag\\
&= (b_{0}, 1-b_{\frac{n+1}{2}}, b_{1}, 1-b_{\frac{n+1}{2}+1}, \cdots, b_{n-1}, 1-b_{\frac{n+1}{2}+n-1}(=1+b_{\frac{n-1}{2}})).
\end{align*}
Namely, for $i=n\mu+(1-n)\lambda\pmod {2n}$, $0\leq\mu\leq1$, $0\leq\lambda\leq n-1$, we have
\begin{align*}
 d_{i}
 &= \left\{ \begin{array}{ll}
b_{\lambda(\frac{1-n}{2})}, & \textrm{if $\mu=0$};\\
1-b_{\lambda(\frac{1-n}{2})}, & \textrm{if $\mu=1$}.
\end{array} \right.
\end{align*}

Finally, we use the Gray mapping
$$\phi: \mathbb{Z}_{2}\oplus \mathbb{Z}_{2}\cong \mathbb{Z}_{4},$$
$$\phi(0, 0)=0,~ \phi(0, 1)=1,~ \phi(1, 1)=2, ~\phi(1, 0)=3$$
to construct a quaternary sequence with period $2n$
\begin{align}\label{E2}
w=w(a, b)=\{w_i\}_{i=0}^{2n-1}, ~~w_{i}=\phi(c_{i}, d_{i}).
\end{align}

Tang and Ding \cite{X. Tang} proved that for any binary sequences $a$ and $b$ with ideal autocorrelation and the same odd period $n$, the interleaved quaternary sequence $w=w(a, b)$ has optimal autocorrelation $A_{w}(\tau)\in \{0, -2\}$ for all $1\leq \tau\leq 2n-1$. Let $W(4)=\sum_{i=0}^{2n-1}w_{i}4^{i}$. In the remaining part of this section we present a general formula on the value $W(4)\pmod{4^{2n}-1}$ for any binary sequences $a$ and $b$. Then we restrict to binary sequences $a$ and $b$ having ideal autocorrelation and give a general result on the 4-adic complexity of the quaternary sequences $w=w(a, b)$.

\begin{lemma}\label{lem1}
Let $n\geq3$ be an odd integer, $a=\{a_\lambda\}_{\lambda=0}^{n-1}$ and $b=\{b_\lambda\}_{\lambda=0}^{n-1}$ be binary sequences with period $n$, $a_\lambda, b_\lambda\in \mathbb{F}_{2}=\{0, 1\}$, $w=w(a, b)=\{w_{i}\}_{i=0}^{2n-1}$, $w_{i}\in\{0, 1, 2, 3\}$ be the interleaved quaternary sequence given by (\ref{E2}), $W(4)=\sum_{i=0}^{2n-1}w_{i}4^{i}\in\mathbb{Z}$. Then
\begin{align*}
 W(4)& \equiv2(4^{n}-1)\sum_{\lambda=0}^{n-1}a_{\lambda}b_{\lambda}4^{2\lambda}+
(4^{n}+3)\sum_{\lambda=0}^{n-1}a_{\lambda}4^{2\lambda}+
(1-4^{n})\sum_{\lambda=0}^{n-1}b_{\lambda}4^{2\lambda}+4\cdot\frac{4^{2n}-1}{15}\pmod{4^{2n}-1}\\
& \equiv\left\{ \begin{array}{ll}
-4\sum_{\lambda=0}^{n-1}a_{\lambda}b_{\lambda}4^{2\lambda}+
2\sum_{\lambda=0}^{n-1}(a_{\lambda}+b_{\lambda})4^{2\lambda}-\frac{4^{n}+1}{5}\pmod{4^{n}+1}\\
4\sum_{\lambda=0}^{n-1}a_{\lambda}4^{2\lambda}+\frac{4^{n}-1}{3}\pmod{4^{n}-1}.
\end{array} \right.
\end{align*}
\end{lemma}

\begin{proof}
From the definition (\ref{E2}) of $w=w(a, b)$,
\begin{align*}
W(4) =\sum_{i=0}^{2n-1}w_{i}4^{i}&\equiv \sum_{\lambda=0}^{n-1}\sum_{\mu=0}^{1}\phi(\mu, \lambda)4^{n\mu+(1-n)\lambda}\pmod{4^{2n}-1}\\
     &\equiv\sum_{\lambda=0}^{n-1}\sum_{\mu=0}^{1}4^{n\mu+
(1-n)\lambda}\left(\sum_{\substack{c_{i}=0\\ d_{i}=1}}1+2\sum_{\substack{c_{i}=1\\ d_{i}=1}}1+3\sum_{\substack{c_{i}=1\\ d_{i}=0}}1\right)\pmod{4^{2n}-1}\notag\\
     &~~~(\text{where} ~i=n\mu+(1-n)\lambda)\notag\\
     &\equiv\sum_{\lambda=0}^{n-1}\sum_{\mu=0}^{1}4^{n\mu+
(1-n)\lambda}\left(2\sum_{c_{i}=1}1+\sum_{c_{i}+d_{i}\equiv1(\bmod{2})}1\right)\pmod{4^{2n}-1}\notag\\
     &\equiv\sum_{\lambda=0}^{n-1}4^{(1-n)\lambda}\left(2\sum_{a_{\lambda(\frac{1-n}{2})}=1}1+
\sum_{a_{\lambda(\frac{1-n}{2})}+b_{\lambda(\frac{1-n}{2})}=1}1\right)\notag\\
     &~~~+\sum_{\lambda=0}^{n-1}4^{n+(1-n)\lambda}\left(2\sum_{a_{\lambda(\frac{1-n}{2})}=1}1+
\sum_{a_{\lambda(\frac{1-n}{2})}=b_{\lambda(\frac{1-n}{2})}}1\right)\pmod{4^{2n}-1}\notag\\
     &\equiv\sum_{\lambda=0}^{n-1}4^{(1-n)\lambda}\left(2a_{\lambda(\frac{1-n}{2})}
+a_{\lambda(\frac{1-n}{2})}(1-b_{\lambda(\frac{1-n}{2})})
+(1-a_{\lambda(\frac{1-n}{2})})b_{\lambda(\frac{1-n}{2})}\right)\notag\\
     &~~~+4^{n}\sum_{\lambda=0}^{n-1}4^{(1-n)\lambda}\left(2a_{\lambda(\frac{1-n}{2})}
+a_{\lambda(\frac{1-n}{2})}b_{\lambda(\frac{1-n}{2})}
+(1-a_{\lambda(\frac{1-n}{2})})(1-b_{\lambda(\frac{1-n}{2})})\right)\pmod{4^{2n}-1}\notag\\
     &\equiv\sum_{\lambda=0}^{n-1}4^{2\lambda}\left(2a_{\lambda}
+a_{\lambda}(1-b_{\lambda})
+(1-a_{\lambda})b_{\lambda}\right)\notag\\
     &~~~+4^{n}\sum_{\lambda=0}^{n-1}4^{2\lambda}\left(2a_{\lambda}+a_{\lambda}b_{\lambda}
+1-a_{\lambda}-b_{\lambda}+a_{\lambda}b_{\lambda}\right)\pmod{4^{2n}-1}\notag\\
     &~~~\left(\text{since} ~\gcd\left(\frac{1-n}{2}, n\right)=1\right)\notag\\
     &\equiv\sum_{\lambda=0}^{n-1}4^{2\lambda}\left(3a_{\lambda}+b_{\lambda}-2a_{\lambda}b_{\lambda}\right)
+4^{n}\sum_{\lambda=0}^{n-1}4^{2\lambda}\left(a_{\lambda}-b_{\lambda}
+2a_{\lambda}b_{\lambda}+1\right)\pmod{4^{2n}-1}\notag\\
     &\equiv2(4^{n}-1)\sum_{\lambda=0}^{n-1}a_{\lambda}b_{\lambda}4^{2\lambda}
+(4^{n}+3)\sum_{\lambda=0}^{n-1}a_{\lambda}4^{2\lambda}
+(1-4^{n})\sum_{\lambda=0}^{n-1}b_{\lambda}4^{2\lambda}
+4^{n}\sum_{\lambda=0}^{n-1}4^{2\lambda}\pmod{4^{2n}-1}.\notag
\end{align*}

From $n$ is odd and $4^{n}\equiv4\pmod {15}$, we have
$$4^{n}\sum_{\lambda=0}^{n-1}4^{2\lambda}\equiv4^{n}\cdot \frac{4^{2n}-1}{15}\equiv4\cdot \frac{4^{2n}-1}{15}\pmod{4^{2n}-1}$$
and
\begin{align*}
 W(4)& \equiv2(4^{n}-1)\sum_{\lambda=0}^{n-1}a_{\lambda}b_{\lambda}4^{2\lambda}+
(4^{n}+3)\sum_{\lambda=0}^{n-1}a_{\lambda}4^{2\lambda}+
(1-4^{n})\sum_{\lambda=0}^{n-1}b_{\lambda}4^{2\lambda}+4\cdot\frac{4^{2n}-1}{15}\pmod{4^{2n}-1}\\
& \equiv \left\{ \begin{array}{ll}
-4\sum_{\lambda=0}^{n-1}a_{\lambda}b_{\lambda}4^{2\lambda}+
2\sum_{\lambda=0}^{n-1}(a_{\lambda}+b_{\lambda})4^{2\lambda}+4\cdot\frac{4^{2n}-1}{15}\pmod{4^{n}+1}\\
4\sum_{\lambda=0}^{n-1}a_{\lambda}4^{2\lambda}+4\cdot\frac{4^{2n}-1}{15}\pmod{4^{n}-1}.
\end{array} \right.
\end{align*}

From $4^n\equiv 4\pmod 5$ we have $4\cdot \frac{4^n-1}{3}\equiv -1\pmod 5$. Then we obtain
$$4\cdot\frac{4^{2n}-1}{15}=\frac{4(4^{n}-1)}{3}\cdot\frac{4^{n}+1}{5}\equiv-\frac{4^{n}+1}{5}\pmod{4^{n}+1}.$$
By $4^n\equiv 1\pmod 3$ we know $4\cdot\frac{4^n+1}{5}\equiv 1\pmod 3$. Then we have $$4\cdot\frac{4^{2n}-1}{15}=\frac{4(4^{n}+1)}{5}\cdot\frac{4^{n}-1}{3}\equiv\frac{4^{n}-1}{3}\pmod{4^{n}-1}.$$

This completes the proof of Lemma \ref{lem1}.
\end{proof}

\begin{remark}
From Lemma \ref{lem1} we know that $W(4)\pmod{4^{n}-1}$ is independent of the sequence $b$.
\end{remark}

Let $s_{\lambda}=(-1)^{a_{\lambda}}$, $t_{\lambda}=(-1)^{b_{\lambda}}$. The value $W(4)=\sum_{\lambda=0}^{2n-1}w_{\lambda}4^{\lambda}\pmod {4^{2n}-1}$ can be expressed in terms of the $\{\pm1\}$-sequences $\{s_{\lambda}\}_{\lambda=0}^{n-1}$ and $\{t_{\lambda}\}_{\lambda=0}^{n-1}$.

\begin{corollary}\label{co2}
Let $a=\{a_\lambda\}_{\lambda=0}^{n-1}$ and $b=\{b_\lambda\}_{\lambda=0}^{n-1}$ be binary sequences with odd period $n$, $a_\lambda, b_\lambda\in\{0, 1\}$, $s_{\lambda}=(-1)^{a_{\lambda}}$, $t_{\lambda}=(-1)^{b_{\lambda}}$. Then for quaternary sequence $w=w(a, b)=\{w_{i}\}_{i=0}^{2n-1}$, $w_{i}\in\{0, 1, 2, 3\}$, $W(4)=\sum_{i=0}^{2n-1}w_{i}4^{i}\in\mathbb{Z}$, we have
\begin{align*}
 W(4)\equiv \left\{ \begin{array}{ll}
-\sum_{\lambda=0}^{n-1}s_{\lambda}t_{\lambda}4^{2\lambda}\pmod{4^{n}+1}\\
-2\sum_{\lambda=0}^{n-1}s_{\lambda}4^{2\lambda}\pmod{4^{n}-1}.
\end{array} \right.
\end{align*}
\end{corollary}

\begin{proof}
We have $a_\lambda=\frac{1}{2}(1-s_{\lambda})$, $b_\lambda=\frac{1}{2}(1-t_{\lambda})$. Then by Lemma \ref{lem1},
\begin{align*}
W(4)& \equiv -4\sum_{\lambda=0}^{n-1}a_{\lambda}b_{\lambda}4^{2\lambda}+
2\sum_{\lambda=0}^{n-1}(a_{\lambda}+b_{\lambda})4^{2\lambda}-\frac{4^{n}+1}{5}\pmod{4^{n}+1}\\
& \equiv -\sum_{\lambda=0}^{n-1}(1-s_{\lambda})(1-t_{\lambda})4^{2\lambda}
+\sum_{\lambda=0}^{n-1}(2-s_{\lambda}-t_{\lambda})4^{2\lambda}-\frac{4^{n}+1}{5}\pmod{4^{n}+1}\notag\\
& \equiv \sum_{\lambda=0}^{n-1}4^{2\lambda}
-\sum_{\lambda=0}^{n-1}s_{\lambda}t_{\lambda}4^{2\lambda}-\frac{4^{n}+1}{5}\pmod{4^{n}+1}\notag\\
& \equiv \frac{4^{2n}-1}{15}-\sum_{\lambda=0}^{n-1}s_{\lambda}t_{\lambda}4^{2\lambda}-\frac{4^{n}+1}{5}
\equiv-\sum_{\lambda=0}^{n-1}s_{\lambda}t_{\lambda}4^{2\lambda}\pmod{4^{n}+1}\notag
\end{align*}
and
\begin{align*}
 W(4)& \equiv4\sum_{\lambda=0}^{n-1}a_{\lambda}4^{2\lambda}+\frac{4^{n}-1}{3}\equiv
2\sum_{\lambda=0}^{n-1}(1-s_{\lambda})4^{2\lambda}+\frac{4^{n}-1}{3}\pmod{4^{n}-1}\\
& \equiv 2\cdot\frac{4^{2n}-1}{15}-2\sum_{\lambda=0}^{n-1}s_{\lambda}4^{2\lambda}
+\frac{4^{n}-1}{3}
\equiv-2\sum_{\lambda=0}^{n-1}s_{\lambda}4^{2\lambda}\pmod{4^{n}-1}.\notag
\end{align*}

\end{proof}

Now we present a general formula on the 4-adic complexity of the quaternary interleaved sequence $w=w(a, b)=\{w_{i}\}_{i=0}^{2n-1}$, where $a$, $b$ are binary sequences with period $n$$(\equiv3(\bmod4))$ and ideal autocorrelation.

From now on we denote by $\sum(n)$ the set of binary sequences with period $n$$(\equiv3(\bmod4))$ and ideal autocorrelation. By the definition of ideal autocorrelation, for $a=\{a_\lambda\}_{\lambda=0}^{n-1}\in\sum(n)$, $a_\lambda\in\{0, 1\}$, $s_{\lambda}=(-1)^{a_{\lambda}}\in\{\pm1\}$, we have
$$A_a(\tau)=\sum_{\lambda=0}^{n-1}(-1)^{a_\lambda+a_{\lambda+\tau}}
=\sum_{\lambda=0}^{n-1}s_{\lambda}s_{\lambda+\tau}=-1 \ (1\leq\tau\leq n-1).$$

\begin{theorem}\label{th33}
Let $n\equiv3\pmod4$, $a=\{a_\lambda\}_{\lambda=0}^{n-1}$ and $b=\{b_\lambda\}_{\lambda=0}^{n-1}$ be known binary sequences in $\sum(n)$, $s_{\lambda}=(-1)^{a_{\lambda}}$, $t_{\lambda}=(-1)^{b_{\lambda}}$. Then the 4-adic complexity of the quaternary interleaved sequence $w=w(a, b)$ is
$$FC_{w}(4)=\log_{4}\left((4^{n}-1)\cdot\frac{4^{n}+1}{d_{+}}\right)$$
where $d_{+}=\gcd\left(\sum_{\lambda=0}^{n-1}s_{\lambda}t_{\lambda}4^{2\lambda}, 4^{n}+1\right)$.

\end{theorem}

Before proving Theorem \ref{th33}, we need the following result.

\begin{lemma}\label{lem4}
Let $n\equiv3\pmod4$ be the period of known binary sequences in $\sum(n)$, $n>3$. Namely, $n=2^{m}-1$ $(m\geq2)$, $n$ is a prime, and $n=pq$ where $p$ and $q=p+2$ are prime numbers. Then $n+1$ and $4^{n}-1$ have no common prime divisor $\pi>3$.
\end{lemma}

\begin{proof}
{\rm (A)} For $n=2^{m}-1$ $(m\geq2)$, we have $\gcd(n+1, 4^{n}-1)=\gcd(2^{m}, 4^{n}-1)=1$.

{\rm (B)} Let $n=p$ be a prime, $p\equiv3\pmod4$, $p>3$. Suppose that $p+1$ and $4^{p}-1$ have a common (odd) prime divisor $\pi>3$. Then $\pi\geq7$, $\pi\mid p+1$ so that $\pi\leq\frac{p+1}{4}$. From $4^{p}\equiv1\pmod \pi$ and $\pi\neq3$ we know that the order of $4\pmod \pi$ is $p$. Therefore $p\mid\pi-1$ by Fermat's Theorem which implies $p\leq(\pi-1)$. Then we get a contradiction $4\pi\leq p+1\leq1+(\pi-1)=\pi$.

{\rm (C)} Let $n=p(p+2)$, a product of twin-primes. Suppose that $n+1=(p+1)^{2}$ and $4^{n}-1$ have a common prime divisor $\pi>3$. Then $4^{p(p+2)}\equiv1\pmod \pi$. From $\pi\neq3$ we know that the order $r$ of $4\pmod \pi$ is at least $p$. We obtain $p\leq r\mid \frac{\pi-1}{2}$. From $\pi\mid n+1=(p+1)^{2}$ we have $\pi\mid p+1$ and $\pi\leq p+1$. Then we obtain a contradiction $p\leq\frac{\pi-1}{2}\leq\frac{p+1-1}{2}=\frac{p}{2}$. This completes the proof of Lemma \ref{lem4}.

\end{proof}

{\bfseries{Proof of Theorem \ref{th33}}}
\ By Definition 1, the 4-adic complexity of $w=w(a, b)$ is
$$FC_{w}(4)=\log_{4}\left(\frac{4^{2n}-1}{d}\right),~ d=\gcd(W(4), 4^{2n}-1),~ W(4)=\sum_{i=0}^{2n-1}w_{i}4^{i}\in\mathbb{Z}.$$
Let
$$d_{+}=\gcd(W(4), 4^{n}+1),~ d_{-}=\gcd(W(4), 4^{n}-1).$$
From $\gcd(4^{n}+1, 4^{n}-1)=1$ we know that $\gcd(d_{+}, d_{-})=1$ and $d=d_{+}d_{-}$.

Firstly we show $d_{-}=1$. Suppose that $d_{-}=\gcd(W(4), 4^{n}-1)>1$. Let $\pi$ be a common (odd) prime divisor of $W(4)$ and $4^{n}-1$. If $\pi=3$, by Corollary \ref{co2}, $0\equiv W(4)\equiv-2\sum_{\lambda=0}^{n-1}s_{\lambda}4^{2\lambda}
\equiv\sum_{\lambda=0}^{n-1}s_{\lambda}\pmod 3$ and then we obtain a contradiction
$$0\equiv\left(\sum_{\lambda=0}^{n-1}s_{\lambda}\right)^{2}=\sum_{\lambda, \mu=0}^{n-1}s_{\lambda}s_{\mu}
=\sum_{\tau=0}^{n-1}\sum_{\lambda=0}^{n-1}s_{\lambda}s_{\lambda+\tau}=n+\sum_{\tau=1}^{n-1}(-1)=1\pmod3.$$
Therefore $\pi\neq3$. From $4^{n}-1=(-1)^{n}-1\equiv3\pmod5$ we know that $\pi\neq5$. Therefore $\pi\geq7$. By $0\equiv W(4)\equiv-2\sum_{\lambda=0}^{n-1}s_{\lambda}4^{2\lambda}
\pmod \pi$ we have
\begin{align*}
 0& \equiv \left(\sum_{\lambda=0}^{n-1}s_{\lambda}4^{2\lambda}\right)
 \left(\sum_{\mu=0}^{n-1}s_{\mu}4^{-2\mu}\right)
 \equiv\sum_{\lambda, \mu=0}^{n-1}s_{\lambda}s_{\mu}4^{2(\lambda-\mu)}\pmod\pi\\
& \equiv \sum_{\tau=0}^{n-1}4^{2\tau}\sum_{\mu=0}^{n-1}s_{\mu+\tau}s_{\mu}
= n+\sum_{\tau=1}^{n-1}(-4^{2\tau})\equiv n-\left(\frac{4^{2n}-1}{15}-1\right)\equiv n+1\pmod\pi ~~(\text{since} ~\pi\geq7).\notag
\end{align*}
This implies that $\pi\mid \gcd(n+1, 4^{n}-1)$. But this contradicts to Lemma \ref{lem4}. Therefore $d_{-}=1$ and $d=d_{+}=\gcd(W(4), 4^{n}+1)$. By Corollary \ref{co2}, $W(4)\equiv-\sum_{\lambda=0}^{n-1}s_{\lambda}t_{\lambda}4^{2\lambda}\pmod{4^{n}+1}$. We have $d_{+}=\gcd(\sum_{\lambda=0}^{n-1}s_{\lambda}t_{\lambda}4^{2\lambda}, 4^{n}+1)$. This completes the proof of Theorem \ref{th33}.

Theorem \ref{th33} has the following consequences. Firstly, from $d_{+}=\gcd(\sum_{\lambda=0}^{n-1}s_{\lambda}t_{\lambda}4^{2\lambda}, 4^{n}+1)$, $FC_{w}(4)=\log_{4}\left(\frac{4^{2n}-1}{d_{+}}\right)$ and $\sum_{\lambda=0}^{n-1}s_{\lambda}t_{\lambda}4^{2\lambda}
=\sum_{\lambda=0}^{n-1}t_{\lambda}s_{\lambda}4^{2\lambda}$ we know that
\begin{corollary}\label{co5}
Let $a,b\in\sum(n)$. Then $w(a,b)$ and $w(b,a)$ have the same 4-adic complexity.
\end{corollary}
Next, we present a transformation group on the set $\sum(n)$.
\begin{definition}\label{def2}
For a binary sequence $a=\{a_i\}_{i=0}^{N-1}$ with period N, $a_i\in\{0,1\}$.\\
(1). The complement sequence of $a$ is $C(a)=\overline{a}=\{\overline{a}_i\}_{i=0}^{N-1}$ where $\overline{a}_i=1-a_i\ (0\leq i\leq N-1)$.\\
(2). For $e\in\mathbb{Z}_N$, the e-shift sequence of $a$ is $L^e(a)=\{b_i\}_{i=0}^{N-1}$ where $b_i=L^e(a_i)=a_{i+e}\ (0\leq i\leq N-1)$.\\
(3). For $r\in\mathbb{Z}_N^*$, the r-sample sequence of $a$ is $M_r(a)=\{c_i\}_{i=0}^{N-1}$ where $c_i=M_r(a_i)=a_{ri}\ (0\leq i\leq N-1)$.
\end{definition}
It is easy to see that the sequences $C(a)$, $L^e(a)$ and $M_r(a)$ have the same period $N$, $C^2=L^0=M_1=Id$, $L^{e_1}L^{e_2}=L^{e_1+e_2}$, $M_{r_1}M_{r_2}=M_{r_1 r_2}$. Since
$$M_rL^e(a_i)=M_r(a_{i+e})=a_{ri+re}=L^{re}(a_{ri})=L^{re}M_r(a_{i}),$$
we have $M_rL^e=L^{re}M_r$. Therefore $C$, $L^e$ $(e\in\mathbb{Z}_N)$ and $M_r$ $(r\in\mathbb{Z}_N^*)$ generate a group $G$. Next consequence shows that $G$ is a transformation group on the set $\sum(n)$, and the action of some transforms on related quaternary interleaved sequences keeps the 4-adic complexity.
\begin{corollary}\label{co6}
Let $a, b \in\sum(n)$. Then\\
(1). For any $\sigma\in G, \sigma(a)\in\sum(n).$\\
(2). For $\sigma=L^e$, the quaternary sequences $w(a,b)$ and $\sigma(w)=(\sigma(a),\sigma(b))$ have the same 4-adic complexity.\\
(3). $w(a,b)$ and $w(\overline{a},b)$ have the same 4-adic complexity.
\end{corollary}
\begin{proof}
Since the group $G$ is generated by $C$, $L^e~(e\in\mathbb{Z}_n)$ and $M_r~(r\in\mathbb{Z}_n^*)$. We need to prove (1) for $\sigma=C$, $L^e$ and $M_r$. Let  $a=\{a_i\}_{i=0}^{n-1}$, $b=\{b_i\}_{i=0}^{n-1}$.\\
(1). For each $\tau, 0\leq\tau\leq n-1$,
$$A_{\overline{a}}(\tau)=\sum_{i=0}^{n-1}(-1)^{\overline{a}_i+\overline{a}_{i+\tau}}=\sum_{i=0}^{n-1}(-1)^{a_i+1+a_{i+\tau}+1}=\sum_{i=0}^{n-1}(-1)^{a_{i}+a_{i+\tau}}=A_{a}(\tau).$$
$$A_{L^e(a)}(\tau)=\sum_{i=0}^{n-1}(-1)^{a_{i+e}+a_{i+e+\tau}}=\sum_{j=0}^{n-1}(-1)^{a_{j}+a_{j+\tau}}=A_{a}(\tau).$$
$$A_{M_r(a)}(\tau)=\sum_{i=0}^{n-1}(-1)^{a_{ir}+a_{(i+\tau)r}}=\sum_{i=0}^{n-1}(-1)^{a_{ir}+a_{ir+\tau r}}=A_{a}(\tau r).$$
From these formulas and $\gcd(r,n)=1$ we know that if $a\in\sum(n)$, then $C(a)=\overline{a}$, $L^e(a)$ and $M_r(a)$ belong to $\sum(n)$.

(2) and (3). Let $s_i=(-1)^{a_i}, t_i=(-1)^{b_i}$, and for $\sigma\in G$, $\sigma(s_i)=(-1)^{\sigma(a_i)}.$ Then we have
$$\sum_{i=0}^{n-1}L^e(s_i)L^e(t_i)4^{2i}=\sum_{i=0}^{n-1}s_{i+e}t_{i+e}4^{2i}\equiv 4^{-2e}\sum_{j=0}^{n-1}s_{j}t_{j}4^{2j}\pmod{4^{2n}-1}$$ and
$$\sum_{i=0}^{n-1}(-1)^{\overline{a}_{i}}t_{i}4^{2i}=\sum_{i=0}^{n-1}(-s_{i})t_{i}4^{2i}=-\sum_{i=0}^{n-1}s_{i}t_{i}4^{2i}$$
which imply (2) and (3) by Theorem \ref{th33}.
\end{proof}

\begin{corollary}\label{co7}
For $a\in\sum(n)$, $w=w(a,a)$, we have $FC_w(4)=\log_4(5(4^n-1))$.
\end{corollary}

\begin{proof}
 Let $s_i=(-1)^{a_i}$. Then
 $$\sum_{i=0}^{n-1}s_{i}s_{i}4^{2i}=\sum_{i=0}^{n-1}4^{2i}=\frac{4^{2n}-1}{15}=\frac{4^n+1}{5}\cdot\frac{4^n-1}{3}.$$
 From  $\frac{4^n-1}{3}\equiv\frac{(-1)^n-1}{3}\equiv-2\cdot 3^{-1}\equiv-4\equiv1\pmod 5$, we obtain $$\sum_{i=0}^{n-1}s_{i}s_{i}4^{2i}\equiv\frac{4^n+1}{5}\pmod{4^n+1}.$$
Therefore for $w=w(a,a)$, $d_{+}=\gcd(\sum_{i=0}^{n-1}s_i^2 4^{2i}, 4^{n}+1)=\frac{4^n+1}{5}$ and  $FC_w(4)=\log_4(\frac{4^{2n}-1}{d_{+}})=\log_4(5(4^n-1))$.
\end{proof}

For arbitrary $a, b\in\sum(n)$, in order to determine the exact value of the 4-adic complexity of $w(a,b)$, we need to compute $\sum_{i=0}^{n-1}s_it_i4^{2i}\pmod{4^n+1}$ which is not easy in general case. In Section \ref{sec33}, we will consider several particular series of binary sequences in $\sum(n)$ ($m$-sequences, Legendre, twin-prime and Hall sequences). For each $a$ in such particular series, we determine the 4-adic complexity of $w(b,c)$ where $b=\sigma(a)$, $c=\tau(a)$ for $\sigma, \tau \in G$.
\section{Interleaving By Several Types of Known Sequences}\label{sec33}

In the last section we present a general formula on the 4-adic complexity $FC_w(4)$ of the quaternary sequences $w=w(a,b)$ with period $2n$ interleaved by two binary sequences $a$ and $b$ in $\sum(n)$. Particularly we get a general lower bound $FC_w(4)\geq\log_4(4^n-1)$. In this section we compute the exact value of $FC_w(4)$, $w=w(a,b)$ by taking $a$ and $b$ to be some particular types of known binary sequences in $\sum(n)$. Namely, we consider $a$ and $b$ as $m$-sequences, Legendre, twin-prime and Hall sequences and their complement, shift or sample sequences.

\subsection{m-sequences}

Let $q=2^k\ (k\geq 2)$, $n=q-1$, $\theta$ be a primitive element of $\mathbb{F}_q$, $\mathbb{F}_q^*=\langle\theta\rangle$, and $T:\mathbb{F}_q\rightarrow\mathbb{F}_2$ be the trace mapping. The $m$-sequence with period $n$ is the binary sequence $m=\{m_i\}_{i=0}^{n-1}$ defined by
$$m_i=T(\theta^i)\in\mathbb{F}_2=\{0,1\}\quad(0\leq i\leq n-1).$$
It is well-known that $m\in\sum(n)$.
\begin{theorem}\label{th88}
Let $m=\{m_i\}_{i=0}^{n-1}$ be the $m$-sequence with period $n=2^k-1~(k\geq2), 1\leq l\leq n-1, w=w(m,L^l(m))$. Then $FC_w(4)$ reaches the maximum value $\log_4(4^{2n}-1)$.
\end{theorem}
\begin{proof}
From Theorem \ref{th33} we know that
$$FC_{w}(4)=\log_{4}\left(\frac{4^{2n}-1}{d_{+}}\right),   ~d_{+}=\gcd(M, 4^{n}+1),$$
where
$$M=\sum_{\lambda=0}^{n-1}(-1)^{T(\theta^{\lambda})+T(\theta^{\lambda+l})}4^{2\lambda}=\sum_{\lambda=0}^{n-1}(-1)^{T(\theta^{\lambda}(1
+\theta^l))}4^{2\lambda}.$$
Since the order of $\theta$ is $n$ and $1\leq l\leq n-1$, we know that $\theta^l+1\neq0$ and $\theta^l+1=\theta^r$ for some $r\in\mathbb{Z}$. Then
\begin{equation}\label{e3}
M=\sum_{\lambda=0}^{n-1}(-1)^{T(\theta^{\lambda+r})}4^{2\lambda}\equiv4^{-2r}\sum_{\lambda=0}^{n-1}(-1)^{T(\theta^{\lambda})}4^{2\lambda}\pmod{4^{n}+1}.
\end{equation}
The following result shows that $R=\sum_{\lambda=0}^{n-1}(-1)^{T(\theta^{\lambda})}4^{2\lambda}\in\mathbb{Z}$ has a property like usual Gauss sum.
\begin{lemma}\label{lem9}
Let $\overline{R}=\sum_{\lambda=0}^{n-1}(-1)^{T(\theta^{\lambda})}4^{-2\lambda}$. Then
$$\overline{R}R\equiv n+1-\frac{4^n+1}{5}\pmod{4^n+1}.$$
\end{lemma}
\begin{proof}
\begin{align*}
\overline{R}R&=\sum_{\lambda,\mu=0}^{n-1}(-1)^{T(\theta^{\lambda}+\theta^{\mu})}4^{2(\lambda-\mu)}\\
&\equiv\sum_{\tau=0}^{n-1}4^{2\tau}\sum_{\mu=0}^{n-1}(-1)^{T(\theta^{\mu+\tau}+\theta^{\mu})}\pmod{4^{2n}-1}\\
&\equiv n+\sum_{\tau=1}^{n-1}4^{2\tau}\sum_{\mu=0}^{n-1}(-1)^{T(\theta^{\mu}(1+\theta^{\tau}))}\pmod{4^{2n}-1}.
\end{align*}
For $1\leq\tau\leq n-1, 1+\theta^\tau\in\mathbb{F}_q^*\ (q=n+1=2^k)$. Therefore
$$\sum_{\mu=0}^{n-1}(-1)^{T(\theta^{\mu}(1+\theta^{\tau}))}=\sum_{x\in\mathbb{F}_q^*}(-1)^{T(x(1+\theta^{\tau}))}=\sum_{x\in\mathbb{F}_q^*}(-1)^{T(x)}=-1.$$
Therefore
\begin{align*}
\overline{R}R&\equiv n-\sum_{\tau=1}^{n-1}4^{2\tau}\equiv n-(\frac{4^{2n}-1}{15}-1)\equiv n+1-\frac{4^{2n}-1}{15}\pmod{4^{2n}-1}.
\end{align*}
Then from $\frac{4^n-1}{3}\equiv -2\cdot 3^{-1}\equiv1 \pmod 5$, we obtain $$\overline{R}R\equiv n+1-\frac{4^n+1}{5}\pmod{4^{n}+1}.$$
\end{proof}
We continue to prove Theorem \ref{th88}. From (\ref{e3}) we know that $d_{+}=\gcd(4^{-2r}R,4^{n}+1)$. Firstly we claim that $d_{+}$ has no prime divisor $\pi\neq5$. Assume that $\pi$ is a prime, $\pi\neq5$ and $\pi|d_{+}$. Then $\pi|4^{n}+1$ and $0\equiv d_{+}\equiv 4^{-2r}R\pmod{\pi}$. Therefore $R\equiv0\pmod{\pi}$ and $0\equiv\overline{R}R\equiv n+1\equiv 2^k\pmod{\pi}$ which is a contradiction. Therefore $d_{+}=5^u\ (u\geq0)$. From
$$M\equiv4^{-2r}\sum_{i=0}^{n-1}(-1)^{T(\theta^i)}4^{2i}\equiv\sum_{i=0}^{n-1}(-1)^{T(\theta^i)}=-1+\sum_{x\in\mathbb{F}_q}(-1)^{T(x)}=-1\pmod{5}$$
we have $5\nmid d_{+}=5^u$. Therefore $d_{+}=1$ and $FC_{w}(4)=\log_{4}(4^{2n}-1)$.
\end{proof}

For $s\in\mathbb{Z}_n^*$, $M_s(m)=\{c_i\}_{i=0}^{n-1}$ where $c_i=m_{si}=T(\theta^{si})$. The sequence $M_s(m)$ is also an $m$-sequence with respect to the primitive element $\theta^s$. If $s\not\equiv2^t\pmod{n}\ (0\leq t\leq k-1)$, the sequences $m$ and $M_s(m)$ are not shift-equivalent, and the 4-adic complexity of $w=w(m,M_s(m))$ is $FC_{w}(4)=\log_{4}\left(\frac{4^{2n}-1}{d_{+}}\right)$, $d_{+}=\gcd(M, 4^{n}+1),$ and
$M=\sum_{i=0}^{n-1}(-1)^{T(\theta^i+\theta^{is})}4^{2i}=\sum_{x\in\mathbb{F}_q^*}(-1)^{T(x+x^s)}4^{2i}\in\mathbb{Z}$. It seems that to compute the exponential sum $M$ is not easy in general case.

\subsection{Twin-Prime sequences}

Let $p$ and $q=p+2$ be prime numbers, $n=pq$. The ring $\mathbb{Z}_{n}=\mathbb{Z}/ n\mathbb{Z}$ has a partition $\mathbb{Z}_{pq}=\{0\}\bigcup P \bigcup Q \bigcup \mathbb{Z}_{pq}^{\ast}$ where $P=\{p, 2p, \ldots, (q-1)p\}$, $Q=\{q, 2q, \ldots, (p-1)q\}$,
\begin{eqnarray*}
\mathbb{Z}_{n}^{*}
&=&  \left\{ a (\bmod {n})| \gcd(a, n)=1\right\}\\
&=& \left\{ ip+jq | 1\leq i \leq q-1, 1\leq j \leq p-1\right\}.
\end{eqnarray*}
The twin-prime sequence $t=\{t_{\lambda}\}_{\lambda=0}^{n-1}$ with period $n=pq$ is defined by, for $0\leq\lambda\leq n-1$,
\begin{eqnarray*}\label{sequence}
t_\lambda=\left\{ \begin{array}{ll}
0,                             & \mbox{ if $\lambda=0$}; \\
0, & \mbox{ if $\lambda\in Q$}; \\
1, & \mbox{ if $\lambda\in P$}; \\
\frac{1}{2}(1-(\frac{\lambda}{p})(\frac{\lambda}{q})), & \mbox{ if $\lambda\in \mathbb{Z}_{n}^{*}$ },
\end{array}
\right.
\end{eqnarray*}
where $(\frac{\lambda}{p})$ is the Legendre symbol. Let $T_{\lambda}=(-1)^{t_{\lambda}}$, then
\begin{eqnarray*}\label{sequence}
T_\lambda=\left\{ \begin{array}{ll}
1,                             & \mbox{ if $\lambda=0$}; \\
1, & \mbox{ if $\lambda\in Q$}; \\
-1, & \mbox{ if $\lambda\in P$}; \\
(\frac{\lambda}{p})(\frac{\lambda}{q}), & \mbox{ if $\lambda\in \mathbb{Z}_{n}^{*}$ }.
\end{array}
\right.
\end{eqnarray*}

It is known that $t$ and its complement, shift and sample sequences belong to $\sum(n)$. For $l\in \mathbb{Z}_{n}^{*}$, we have sample sequence $M_l(t)=\{M_l(t_\lambda)\}_{\lambda=0}^{n-1}$ where $M_l(t_\lambda)=t_{l\lambda}$. In fact, there are only two sample sequences of $t$ :
\begin{eqnarray*}\label{sequence}
M_l(t)=\left\{ \begin{array}{ll}
t, & \mbox{ if $(\frac{l}{p})(\frac{l}{q})=1$}; \\
\tau(t), & \mbox{ if $(\frac{l}{p})(\frac{l}{q})=-1$},
\end{array}
\right.
\end{eqnarray*}
where $\tau(t)=\{\tau(t_{\lambda})\}_{\lambda=0}^{n-1}$ and
\begin{eqnarray*}\label{sequence}
\tau(t_{\lambda})=\left\{ \begin{array}{ll}
\frac{1}{2}(1-(\frac{l\lambda}{p})(\frac{l\lambda}{q}))=\frac{1}{2}(1+(\frac{\lambda}{p})(\frac{\lambda}{q})), & \mbox{ if $\lambda\in \mathbb{Z}_{n}^{*}$}; \\
t_{\lambda}, & \mbox{ otherwise}.
\end{array}
\right.
\end{eqnarray*}

Since $(\frac{-1}{p})(\frac{-1}{q})=(-1)^{\frac{p-1}{2}+\frac{q-1}{2}}=(-1)^{\frac{p-1}{2}+\frac{p+1}{2}}=(-1)^p=-1$, we know that
$\tau(t)=M_{-1}(t)$.
Let $\tau(T_{\lambda})=(-1)^{\tau(t_{\lambda})}$, then
\begin{eqnarray*}\label{sequence}
\tau(T_{\lambda})=\left\{ \begin{array}{ll}
-T_{\lambda}, & \mbox{ if $\lambda\in \mathbb{Z}_{n}^{*}$}; \\
T_{\lambda}, & \mbox{ otherwise}.
\end{array}
\right.
\end{eqnarray*}

Now we compute $FC_w(4)$ for $w(a,b)$, $a, b\in\{t, \tau(t)\}$ and $a\neq b$.

Since $w(a,b)$ and $w(b,a)$ have the same 4-adic complexity, we need to compute $FC_w(4)$ where $w=w(a,b)$ for $(a,b)=(t, \tau(t))$.
\begin{theorem}\label{th10}
Let $t$ be the twin-prime sequence with period $n=pq\ (q=p+2)$. Then
$$FC_w(4)=\log_4(4^{2n}-1),\ for\ w=w(t, \tau(t)).$$
\end{theorem}
\begin{proof}
Let $t=\{t_{\lambda}\}_{\lambda=0}^{n-1}$, $\tau(t)=\{\tau({t_\lambda})\}_{\lambda=0}^{n-1}$, $T_{\lambda}=(-1)^{t_{\lambda}}$ and $\tau(T_{\lambda})=(-1)^{\tau(t_{\lambda})}$. Then by Theorem \ref{th33} for $w=w(t, \tau(t))$, we have $FC_w(4)=\log_4(\frac{4^{2n}-1}{d})$ where $d=\gcd(S_w(4), 4^n+1)$, $S_w(4)=\sum_{\lambda=0}^{n-1}T_{\lambda}\tau(T_{\lambda})4^{2\lambda}$.

It is easy to see that $5|4^n+1$, $5|4^p+1$ and $5|4^q+1$. By
$$\gcd(4^p+1,4^q+1)=\gcd(4^p+1,4^2\cdot4^p+1)=\gcd(4^p+1,15)=5,$$
we know that $N=\frac{(4^n+1)\cdot5}{(4^p+1)(4^q+1)}\in \mathbb{Z}$ and $4^n+1=N\cdot\frac{4^p+1}{5}\cdot\frac{4^q+1}{5}\cdot5$. Therefore
\begin{eqnarray}\label{e4}
d=\gcd(S_w(4), 4^n+1)|\gcd(S_w(4), 5)\cdot\gcd(S_w(4), N) \cdot\gcd(S_w(4), \frac{4^p+1}{5}) \cdot\gcd(S_w(4), \frac{4^q+1}{5}).
\end{eqnarray}
For $w=w(t, \tau(t))$,
\begin{eqnarray*}
&S_w(4)&\\
&=&\sum_{\lambda=0}^{n-1}T_{\lambda}\tau(T_{\lambda})4^{2\lambda}\\
&=&\sum_{\lambda\in{\mathbb{Z}_n\backslash\mathbb{Z}_n^*}}4^{2\lambda}-\sum_{\lambda\in{\mathbb{Z}_n^*}}4^{2\lambda}\\
&=&2\sum_{\lambda\in{\{0\}\bigcup P \bigcup Q}}4^{2\lambda}-\sum_{\lambda\in{\mathbb{Z}_n}}4^{2\lambda}\\
&=&2(\sum_{i=0}^{p-1}4^{2iq}+\sum_{j=0}^{q-1}4^{2jp}-1)-\frac{4^{2n}-1}{15}\\
&=&2\left(\frac{4^{2n}-1}{4^{2p}-1}+\frac{4^{2n}-1}{4^{2q}-1}-1\right)-\frac{4^{2n}-1}{15}\\
&\equiv&2\left(\frac{4^{n}+1}{4^{p}+1}\cdot\frac{4^{n}-1}{4^{p}-1}+\frac{4^{n}+1}{4^{q}+1}\cdot\frac{4^{n}-1}{4^{q}-1}-1\right)-\frac{4^n+1}{5}\pmod{4^n+1}\ ({\rm{since}}\ \frac{4^n-1}{3}\equiv 1\pmod 5)\\
&\equiv&\left\{ \begin{array}{ll}
2\left(\frac{N(4^q+1)}{5}\cdot\frac{4^n-1}{4^p-1}+\frac{N(4^p+1)}{5}\cdot\frac{4^n-1}{4^q-1}-1\right)\pmod{N}\\
2\left(\sum_{t=0}^{q-1}(-4^p)^t\cdot\sum_{t=0}^{q-1}4^{pt}+N\cdot\frac{4^p+1}{5}\cdot\frac{4^n-1}{4^q-1}-1\right)\pmod{\frac{4^p+1}{5}}\\
2\left(\frac{N(4^q+1)}{5}\cdot\frac{4^n-1}{4^p-1}+\sum_{t=0}^{p-1}(-4^q)^t\cdot\sum_{t=0}^{p-1}4^{qt}-1\right)\pmod{\frac{4^q+1}{5}}
\end{array}
\right.\\
&\equiv&\left\{ \begin{array}{ll}
-2\not\equiv0\pmod{N}~~~~~~N=\frac{(4^n+1)\cdot5}{(4^p+1)(4^q+1)}=\frac{4^n+1}{5}/(\frac{4^p+1}{5}\cdot\frac{4^q+1}{5})\\
2(q+0-1)=2(p+1)\pmod{\frac{4^p+1}{5}}\\
2(0+p-1)=2(p-1)\pmod{\frac{4^q+1}{5}}
\end{array}
\right.
\end{eqnarray*}
and
\begin{eqnarray*}
S_w(4)\equiv |\mathbb{Z}_n\backslash \mathbb{Z}_n^*|-|\mathbb{Z}_n^*|=n-2|\mathbb{Z}_n^*|=pq-2(p-1)(p+1)\equiv -(p^2-2p-2)\pmod{5}.
\end{eqnarray*}

If $5|S_w(4)$ then $(p-1)^2=p^2-2p+1\equiv 3\pmod{5}$ which contradicts to $(\frac{3}{5})=-1$. Therefore $\gcd(S_w(4),$ $5)=1$. Assume that $\pi|\gcd(S_w(4),\ \frac{4^p+1}{5})=\gcd(2(p+1),\frac{4^p+1}{5})$. By Fermat's Theorem, we have $4^{\pi-1}\equiv 1\pmod {\pi}$. From $4^p\equiv -1\pmod {\pi}$ and $\pi\neq 3,5$ we know the order of $4 \pmod \pi$ is $2p$ which implies $2p|\pi-1$. This contradicts to $\pi|p+1$. Therefore $\gcd(S_w(4), \frac{4^p+1}{5})=\gcd(2(p+1),\frac{4^p+1}{5})=1$. By similar argument we can obtain $\gcd(S_w(4), \frac{4^q+1}{5})=\gcd((p-1),\frac{4^q+1}{5})=1$. By (\ref{e4}) and Theorem 3 we have $FC_w(4)=\log_4(4^{2n}-1)$ for $w=(t, \tau(t)).$
\end{proof}
\subsection{Legendre sequences}

Let $p\equiv3\pmod{4}$ be a prime number. The Legendre sequences with period $p$ are $\ell=(\ell_{\lambda})_{\lambda=0}^{p-1}$ and $\ell'=(\ell_{\lambda}')_{\lambda=0}^{p-1}$ defined by
$$\ell_0=0, \ell'_0=1, \ell_{\lambda}=\ell_{\lambda}'=\frac{1}{2}\left(1-\left(\frac{\lambda}{p}\right)\right)\in\{0,1\}\quad for\ 1\leq\lambda\leq p-1.$$
Let $L_{\lambda}=(-1)^{\ell_{\lambda}}, L_{\lambda}'=(-1)^{\ell'_{\lambda}}$, then
$$L_0=1, L'_0=-1, L_{\lambda}=L_{\lambda}'=\left(\frac{\lambda}{p}\right)\quad for\ 1\leq\lambda\leq p-1.$$
It is known that $\ell, \ell'$ and their complement, shift and sample sequences belong to $\sum(p)$. $\ell$ has only two sample sequences $\ell$ and $\tau(\ell)=M_{-1}(\ell)$, $\ell'$ has only two sample sequences $\ell'$ and $\tau(\ell')=M_{-1}(\ell')$. Let $\tau(L_i)=(-1)^{\tau(\ell_i)}$ and
$\tau(L_i')=(-1)^{\tau(\ell_i')}$. Then
$$\tau(L_0)=1, \tau(L_0')=-1, \tau(L_{\lambda})=-L_{\lambda}, \tau(L_{\lambda}')=-L_{\lambda}'\quad (1\leq\lambda\leq p-1).$$
Like in the twin-prime sequences case, for $a,b\in\{\ell, \ell', \tau(\ell), \tau(\ell')\}$, $w(a,b)$, $w(b,a)$, and $w(\tau(a),\tau(b))$ have the same 4-adic complexity.

\begin{theorem}\label{th11}
Let $p\equiv3\pmod{4}$ be a prime number, $\ell$ and $\ell'$ be the Legendre sequences with period $p$. Then
\begin{align*}
FC_w(4)=\left\{ \begin{array}{ll}
\log_4(4^{2p}-1),&\ if\ 5\nmid p-2  \ and\ w\in\{w(\ell,\ell'), w(\ell,\tau(\ell)), w(\ell',\tau(\ell'))\} \\
\log_4(\frac{4^{2p}-1}{5}),&\ if\ 5\mid p-2  \ and\ w\in\{w(\ell,\ell'), w(\ell,\tau(\ell)), w(\ell',\tau(\ell'))\} \\
\log_4(5(4^{p}-1)),&\ if\ w=w(\ell,\tau(\ell')).
\end{array}
\right.
\end{align*}
\end{theorem}
\begin{proof}
(1). For $w=w(\ell,\ell')$, since
$$4^{p}+1=(5-1)^{p}+1\equiv-1+5p+1\equiv5p\pmod{25}$$ and $p\equiv3\pmod{4}$ we have $5\|4^p+1~(5|4^p+1$, $5^2\nmid4^p+1)$ and
\begin{align*}
S_w(4)=\sum_{\lambda=0}^{p-1}L_{\lambda}L_{\lambda}'4^{2\lambda}=-1+\sum_{\lambda=1}^{p-1}4^{2\lambda}=\frac{4^{2p}-1}{15}-2
&\equiv\left\{ \begin{array}{ll}
-2\pmod{\frac{4^p+1}{5}}\\
p-2\pmod{5}
\end{array}
\right..
\end{align*}
If $5\nmid p-2$ , then $\gcd(S_w(4),4^p+1)=\gcd(S_w(4),\frac{4^p+1}{5})\cdot\gcd(S_w(4),5)=1$. Therefore $FC_w(4)=\log_4(4^{2p}-1).$ For $5\mid p-2$, $\gcd(S_w(4),4^p+1)=1\cdot5=5$, we have $FC_w(4)=\log_4(\frac{4^{2p}-1}{5}).$

(2). For $w=w(\ell,\tau(\ell))$ and $w(\ell',\tau(\ell'))$,
$$S_w(4)=\sum_{\lambda=0}^{p-1}L_{\lambda}\tau(L_{\lambda})4^{2\lambda}
=\sum_{\lambda=0}^{p-1}L_{\lambda}'\tau(L_{\lambda}')4^{2\lambda}=1-\sum_{\lambda=1}^{p-1}4^{2\lambda}=2-\frac{4^{2p}-1}{15}.$$
Then we obtain the same result as (1).

(3). For $ w=w(\ell,\tau(\ell')$,
$$S_w(4)=\sum_{\lambda=0}^{p-1}L_{\lambda}\tau(L_{\lambda}')4^{2\lambda}=-\sum_{\lambda=0}^{p-1}4^{2\lambda}=-\frac{4^{2p}-1}{15}.$$
From $\frac{4^p-1}{3}\equiv 1\pmod 5$, we know $S_w(4)\equiv-\frac{4^p+1}{5}\pmod {4^p+1}$.
Then we have $d=\gcd(S_w(4),4^p+1)=\frac{4^{p}+1}{5}$ and $FC_w(4)=\log_4\frac{4^{2p}-1}{d}=\log_4(5(4^{p}-1))$ by Theorem \ref{th33}.
\end{proof}

\subsection{Interleaving by Hall sequences}
Let $p=4x^{2}+27$ be a prime number, $x\equiv1\pmod 3$. Then $6\mid p-1$. Let $\mathbb{F}_{p}^{*}=\langle g\rangle $, $D=\langle g^{6}\rangle$ and $D_{\lambda}=g^{\lambda}D~(0\leq\lambda\leq5)$ be the cyclotomic classes of order 6 in $\mathbb{F}_{p}^{*}$. The Hall sequence with period $p$ is the binary sequence $h=\{h_{\lambda}\}_{\lambda=0}^{p-1}$ defined by
\begin{align*}
 h_{\lambda}= \left\{ \begin{array}{ll}
1, & \textrm{if $\lambda\in D_{0}\cup D_{1}\cup D_{3}$;}\\
0, & \textrm{otherwise}.
\end{array} \right.
\end{align*}

It is known that $h$ and complement, shift, sample of $h$ belong to $\sum(p)$ (see \cite{T. storer}, p.73-74, Theorem 18, where it is proved as an equivalent form: $D_{0}\cup D_{1}\cup D_{3}$ is a difference set of the group $(\mathbb{F}_{p}, +)$). For $1\leq r\leq p-1$, there exists six sample sequences of $h$:
$$M_{r}(h)=h^{(l)}=\{h_{\lambda}^{(l)}\}_{\lambda=0}^{p-1} ~~~\text{if}~ r\in D_{l},$$
where for $0\leq\lambda\leq p-1$,
\begin{align*}
 h_{\lambda}^{(l)}= \left\{ \begin{array}{ll}
1, & \textrm{if $\lambda\in D_{l}\cup D_{1+l}\cup D_{3+l}$;}\\
0, & \textrm{otherwise}.
\end{array} \right.
\end{align*}
Particularly, we take $r=-1$. From $-1=g^{\frac{p-1}{2}}$ and $\frac{p-1}{2}=2x^{2}+13\equiv3\pmod 6$, we have $-1\in D_{3}$. Thus $M_{-1}(h)=h^{(3)}=\{h_{\lambda}^{(3)}\}_{\lambda=0}^{p-1}$ where for $0\leq\lambda\leq p-1$,
\begin{align*}
 h_{\lambda}^{(3)}= \left\{ \begin{array}{ll}
1, & \textrm{if $\lambda\in D_{3}\cup D_{4}\cup D_{0}$;}\\
0, & \textrm{otherwise}.
\end{array} \right.
\end{align*}

In order to determine the 4-adic complexity of $w=w(h, M_{r}(h))$, we need some facts on cyclotomic numbers and ``Gauss periods".
\begin{definition}\label{def3}
Let $p\equiv1\pmod6$, $p-1=6f$, $\mathbb{F}_{p}^{*}=\langle g\rangle$, $D=\langle g^{6}\rangle$ and $D_{\lambda}=g^{\lambda}D~(0\leq\lambda\leq5)$ be the cyclotomic classes of order 6.

(1). The ``Gauss periods" of order 6 in $\mathbb{F}_{p}^{*}$ and valued in $\mathbb{Z}_{N}~(N=4^{2p}-1)$ are defined by
$$\xi_{\lambda}=\sum_{i\in D_{\lambda}}4^{2i}\pmod {4^{2p}-1}\ \ (0\leq\lambda\leq5).$$

Similarly we have the cyclotomic classes of order 3 in $\mathbb{F}_{p}^{*}$
$$C_{\lambda}=g^{\lambda}C,\ \ C=\langle g^{3}\rangle~(0\leq\lambda\leq2),\ \ |C_{\lambda}|=\frac{p-1}{3}$$
and ``Gauss periods" of order 3 $\eta_{\lambda}=\sum_{i\in C_{\lambda}}4^{2i}\pmod {4^{2p}-1}$. We have
$$C_{\lambda}=D_{\lambda}\cup D_{\lambda+3},\ \ \eta_{\lambda}=\xi_{\lambda}+\xi_{\lambda+3}.$$

(2). The cyclotomic numbers of order 3 over $\mathbb{F}_{p}$ is defined by
$$(i, j)=|(C_{i}+1)\cap C_{j}|=\{a\in C_{i}\mid a+1\in C_{j}\}\ \ (0\leq i, j\leq2).$$
\end{definition}

Now we show that the 4-adic complexity of $w=w(h, M_{-1}(h))$ can be determined by using cyclotomic numbers and ``Gauss periods" of order 3. (For other sample sequence $M_{r}(h)$ of $h~(r\in \mathbb{Z}_{p}^{*}, r\neq-1)$, in order to determine the 4-adic complexity of $w(h, M_{r}(h))$ we need more information on cyclotomic numbers and ``Gauss periods" of order 6.) We use the following results on cyclotomic numbers $(i, j)~(0\leq i, j\leq2)$ and ``Gauss periods" $\eta_{\lambda}~(0\leq \lambda\leq2)$. It is known that any prime number $p\equiv1\pmod 3$ can be expressed as $4p=c^{2}+27d^{2}$ where $c, d\in \mathbb{Z}$, $c$ is uniquely determined by $c\equiv1\pmod 3$ and $d$ is determined up to sign.
\begin{lemma}\label{lem11}
Let $p\equiv1\pmod 3$ be a prime number, $4p=c^{2}+27d^{2}$, $c\equiv1\pmod 3$.

(1). (\cite{T. storer}, Lemma 7, p.32-35) The cyclotomic numbers $(i, j)~(0\leq i, j\leq2)$ of order 3 over $\mathbb{F}_{p}$ are listed in the following table\\

\begin{minipage}{\textwidth}
 \begin{minipage}[h]{0.4\textwidth}
  \centering
     \begin{tabular}{c|ccc}
		$(i, j)$& $j=0$ & 1 & 2 \\
		\hline
		$i=0$& $A$ & $B$ & $C$ \\
		1& $B$ & $C$ & $D$\\
		 2& $C$ & $D$ &$B$\\
	\end{tabular}
  \end{minipage}
  \begin{minipage}[h]{0.4\textwidth}
   \centering

   \begin{tabular}{cccc}
$$9A=p-8+c$$\\
$$18B=2p-4-c-9d$$\\
$$18C=2p-4-c+9d$$\\
$$9D=p+1+c$$\\
	  \end{tabular}
   \end{minipage}
\end{minipage}\\

(2). The ``Gauss periods" $\eta_{\lambda}$ $(0\leq\lambda\leq2)$ have the following relation
$$\eta_{m}\eta_{m+k}=\sum_{h=0}^{2}(k, h)\eta_{h+m}+\frac{p-1}{3}\delta_{k}\pmod{\frac{4^{2p}-1}{15}}~(m, k\in \{0, 1, 2\})$$
where $\delta_{0}=1$, $\delta_{1}=\delta_{2}=0$.
\end{lemma}

\begin{proof}
(2). By using similar proof given in \cite{T. storer}~(Lemma 8, p.38-39) for usual ``Gauss periods".
\end{proof}

\begin{theorem}\label{th12}
Let $p=4x^{2}+27$ be a prime number, $h=\{h_{\lambda}\}_{\lambda=0}^{p-1}$ be the Hall sequence with period $p$. Then for the quaternary interleaved sequence $w=w(h, M_{-1}(h))$,
\begin{align*}
FC_w(4)=\left\{ \begin{array}{ll}
\log_4(\frac{4^{2p}-1}{5}),& \textrm{if $p\equiv3\pmod5$;}\\
\log_4(4^{2p}-1),& \textrm{otherwise}.
\end{array}
\right.
\end{align*}
\end{theorem}

\begin{proof}
By Theorem \ref{th33} we know that $FC_w(4)=\log_4(\frac{4^{2p}-1}{d_{+}})$, where $d_{+}=\gcd(S_{w}(4), 4^{p}+1)$ and $$S_w(4)\equiv\sum_{\lambda=0}^{p-1}H_{\lambda}M_{-1}(H_{\lambda})4^{2\lambda}\pmod{4^{2p}-1}.$$
By the definition as before, for $0\leq\lambda\leq p-1$,
\begin{align*}
 H_{\lambda}=(-1)^{h_{\lambda}}= \left\{ \begin{array}{ll}
-1, & \textrm{if $\lambda\in D_{0}\cup D_{1}\cup D_{3}$;}\\
1, & \textrm{otherwise}.
\end{array} \right.
  \ \ M_{-1}(H_{\lambda}) = \left\{ \begin{array}{ll}
-1, & \textrm{if $\lambda\in D_{3}\cup D_{4}\cup D_{0}$;}\\
1, & \textrm{otherwise}.
\end{array} \right.
\end{align*}
We have
\begin{align*}
S_w(4)&\equiv\sum_{\lambda\in \mathbb{F}_{p}\backslash D_{1}\cup D_{4} }4^{2\lambda}-\sum_{\lambda\in D_{1}\cup D_{4}}4^{2\lambda}=\sum_{\lambda=0 }^{p-1}4^{2\lambda}-2\sum_{\lambda\in D_{1}\cup D_{4}}4^{2\lambda}\pmod{4^{2p}-1}\\
&\equiv \frac{4^{2p}-1}{15}-2\eta_{1}\equiv\frac{4^{p}+1}{5}-2\eta_{1}\pmod{4^{p}+1}.
\end{align*}

Therefore $d_{+}=\gcd\left(\frac{4^{p}+1}{5}-2\eta_{1}, 4^{p}+1\right)$. Let $d_{+}=5^{l}\cdot d'$, $5\nmid d'$. Firstly we show that $d'=1$. Assume that $d'>1$. Let $\pi$ be a prime divisor of $d'$. Then $\pi>5$, $\pi\mid 4^{p}+1$ and $0\equiv\frac{4^{p}+1}{5}-2\eta_{1}\equiv-2\eta_{1}\pmod\pi$, namely, $\pi\mid\eta_{1}$. By Lemma \ref{lem11} we obtain the following linear equations over the finite field $\mathbb{F}_{\pi}$
$$C\eta_{0}+A\eta_{1}+B\eta_{2}\equiv\eta_{1}^2-\frac{p-1}{3}\equiv-\frac{p-1}{3}\pmod\pi,$$
$$B\eta_{0}+C\eta_{1}+D\eta_{2}\equiv\eta_{0}\eta_{1}\equiv0\pmod\pi,$$
$$D\eta_{0}+B\eta_{1}+C\eta_{2}\equiv\eta_{2}\eta_{1}\equiv0\pmod\pi.$$
To solve this system of linear equations in variables $(\eta_{0}, \eta_{1}, \eta_{2})$, we obtain $M\eta_{1}\equiv M'\pmod\pi$ where $M$ is the determinant of coefficient matrix and
 \begin{align*}
M'&=\left|
      \begin{array}{ccc}
        C & -\frac{p-1}{3} & B \\
        B & 0 & D \\
        D & 0 & C \\
      \end{array}
    \right|
=\frac{p-1}{3}\left|
                 \begin{array}{cc}
                   B & D \\
                   D & C \\
                 \end{array}
               \right|
=\frac{p-1}{3}(BC-D^{2})\\
&=\frac{p-1}{3}\left[\frac{1}{18^{2}}(2p-4-c-9d)(2p-4-c+9d)-\frac{1}{9^{2}}(p+1+c)^{2}\right] \\
&=\frac{p-1}{3^{4}}(1-3p-pc).
\end{align*}

From $\pi\mid \eta_{1}$ we get $\frac{p-1}{3^{4}}(1-3p-pc)\equiv M\eta_{1}\equiv0\pmod\pi$. From $4^{p}\equiv-1\pmod\pi$ and $\pi\neq3,5$ we know that order of $4\pmod\pi$ is $2p$. By Fermat's Theorem, $2p\mid\pi-1$, $\pi-1=2pt$ where $t$ is a positive number and then $\pi>p-1$. Therefore $\pi\nmid p-1$ and $1-3p-pc\equiv0\pmod\pi$. Multiple by $2t$, we obtain
$$0\equiv2t(pc+3p-1)\equiv(\pi-1)(c+3)-2t\equiv-(c+3+2t)\pmod\pi.$$
From $p=4x^{2}+27\geq31$, $\pi=1+2tp>2p=62$, $4p=c^{2}+27d^{2}$ we have
$$|c+3+2t|\leq\sqrt{4p}+3+\frac{\pi-1}{p}\leq\sqrt{2\pi}+3+\frac{\pi}{31}<\pi.$$
By $c+3+2t\equiv0\pmod\pi$ we obtain $c+3+2t=0$. Then by $p\equiv c\equiv1\pmod3$  we have
$$\pi=2tp+1=-p(c+3)+1\equiv-pc+1\equiv0\pmod3$$
which is a contradiction since $\pi$ is a prime and $\pi\neq3$. Therefore $d'=1$ and $d_{+}=5^{l}$. At last we determine $d_{+}$. Firstly,
 by the proof of Theorem \ref{th11} we have $$5\parallel 4^{p}+1,\ \text{if} \ p\neq5;\ \text{and} \ 25\mid 4^{5}+1.$$
Then from $$S_w(4)=\sum_{\lambda=0}^{p-1}4^{2\lambda}-2\sum_{\lambda\in D_{1}\cup D_{4}}4^{2\lambda}\equiv p-2|D_{1}\cup D_{4}|=p-\frac{p-1}{3}=\frac{2p+1}{3}\pmod5$$
we obtain
$$5\mid S_{w}(4)\ \text{if and only if}\ p\equiv2\pmod5.$$
Therefore, if $p\equiv2\pmod5$ then $d_{+}=5$ we have $FC_w(4)=\log_4(\frac{4^{2p}-1}{5})$. Otherwise, $5\nmid S_{w}(4)$, then $d_{+}=1$ and $FC_w(4)=\log_4(4^{2p}-1)$.
\end{proof}
\begin{remark}
To resist the attack of the rational algorithm, the $4$-adic complexity of a quaternary sequence $w(a,b)$ with period $N$ should exceed $\frac{N-16}{6}$. From
Theorem \ref{th33} we can see that the quaternary interleaved sequences defined in  \cite{X. Tang} are safe enough to resist the attack of the rational approximation
algorithm. By Theorem \ref{th88}, \ref{th10}, \ref{th11}, \ref{th12}, we can see that by considering $a$ and $b$ to be several types of known binary sequences with ideal autocorrelation ($m$-sequences, twin-prime sequences, Legendre and Hall sequences and their complement, shift or sample sequences), the $4$-adic complexity of a quaternary sequence $w(a,b)$ reaches or nearly reaches the maximum in most cases.
\end{remark}
\section{Conclusion}

Tang and Ding \cite{X. Tang} present a series of quaternary sequences $w(a, b)$ interleaved by two binary sequences $a$ and $b$ with ideal autocorrelation and show that such interleaved quaternary sequences have optimal autocorrelation. In this paper we consider the 4-adic complexity $FC_{w}(4)$ of such interleaved quaternary sequences $w=w(a, b)$. Firstly we present a general formula on $FC_{w}(4)$, $w=w(a, b)$. As a direct consequence, we obtain lower bound $FC_{w}(4)\geq\log_4(4^n-1)$ where $2n$ is the period of $w=w(a, b)$. Then by considering $a$ and $b$ to be several types of known binary sequences ($m$-sequences, twin-prime sequences, Legendre and Hall sequences and their complement, shift or sample sequences), we obtain the exact value of $FC_{w}(4)$. The results show that in most cases, $FC_{w}(4)$ reaches or nearly reaches the maximum value $\log_4(4^{2n}-1)$. Our results show that the quaternary
sequences defined in \cite{X. Tang} are safe enough from the viewpoint of the 4-adic complexity.
Many constructions on binary sequences with period $n\equiv3 \pmod4$ and ideal autocorrelation have been found. It seems to use more skill in order to determine the exact values of the 4-adic complexity of the quaternary sequences interleaved by other types of binary sequences with ideal autocorrelation.

\ifCLASSOPTIONcaptionsoff
  \newpage
\fi



%


\end{document}